\documentclass[11pt]{article}
\pdfoutput=1

\usepackage[margin=1in]{geometry}
\usepackage[utf8]{inputenc}
\usepackage{amssymb, amsthm}
\usepackage{cite}
\usepackage{hyperref}
\usepackage{authblk}
\usepackage{graphicx}
\usepackage[labelformat=simple]{subcaption}

\usepackage[ruled, vlined, linesnumbered]{algorithm2e}
\usepackage{booktabs}
\usepackage{makecell}
\usepackage{rotating}
\usepackage[export]{adjustbox}
\usepackage{tikz}
\usetikzlibrary{quantikz}
\usepackage{url}

\usepackage{xcolor}
\usepackage{colortbl}

\newcommand{\bs}[1]{\boldsymbol #1}

\newtheorem{theorem}{Theorem}

\title{Quantum binary field multiplication with subquadratic\\ Toffoli gate count and low space-time cost}

\author[1,2]{Vivien Vandaele}
\affil[1]{Eviden Quantum Lab, Les Clayes-sous-Bois, France}
\affil[2]{Université de Lorraine, CNRS, Inria, LORIA, F-54000 Nancy, France}

\date{}

\begin{document}
\maketitle

\begin{abstract}
    Multiplication over binary fields is a crucial operation in quantum algorithms designed to solve the discrete logarithm problem for elliptic curve defined over $GF(2^n)$.
    In this paper, we present an algorithm for constructing quantum circuits that perform multiplication over $GF(2^n)$ with $\mathcal{O}(n^{\log_2(3)})$ Toffoli gates.
    We propose a variant of our construction that achieves linear depth by using $\mathcal{O}(n\log(n))$ ancillary qubits.
    This approach provides the best known space-time trade-off for binary field multiplication with a subquadratic number of Toffoli gates.
    Additionally, we demonstrate that for some particular families of primitive polynomials, such as trinomials, the multiplication can be done in logarithmic depth and with $\mathcal{O}(n^{\log_2(3)})$ gates.
\end{abstract}

\section{Introduction}

Binary field arithmetic plays a crucial role in the design of quantum polynomial-time attacks on elliptic curve cryptography over binary fields \cite{Shor, Proos_2003, cheung2008design}.
Consequently, much work has been devoted to analyzing and optimizing the arithmetic operations required for solving the discrete logarithm problem on binary elliptic curves \cite{Kaye_2005, Amento_2013, Budhathoki_2014, Banegas_2020, putranto2022another, Kim_2023, Taguchi_2024}.
Among these operations, multiplication of elements in binary fields stands out as one of the most computationally expensive.
As such, various techniques have been developed to construct efficient quantum circuits performing binary field multiplication.
In Reference~\cite{cheung2008design}, a linear depth circuit composed of $\mathcal{O}(n^2)$ Toffoli gates is presented, based on the classical Mastrovito multiplier \cite{Mastrovito_1989, mastrovito1991vlsi}.
It was demonstrated in Reference~\cite{Rotteler_2014} that the depth of the circuit could be reduced to $\mathcal{O}(\log_2(n))$, at the cost of $\mathcal{O}(n^2)$ ancillary qubits.
Subsequent research has focused on reducing the number of Toffoli gates, as they are particularly costly to implement fault-tolerantly.
Notable progress has been made using approaches inspired by the classical Karatsuba multiplication algorithm~\cite{karatsuba1962multiplication}, which led to a Toffoli gate count of $\mathcal{O}(n^{\log_2(3)})$~\cite{Kepley_2015, van_Hoof_2020, Putranto_2023}.
However, this subquadratic Toffoli gate count comes with trade-offs: either in space, by requiring significantly more qubits, or in time, by substantially increasing the circuit depth.
As a result, despite achieving a subquadratic number of Toffoli gates, these circuits introduce a significant overhead in space-time cost.
Notably, the space-time costs of these circuits, despite having a subquadratic number of Toffoli gates, exceed those of the circuits containing a quadratic number of Toffoli gates proposed in References~\cite{cheung2008design, Rotteler_2014}

In this paper, we propose a novel approach, also similar to the Karatsuba multiplication algorithm, to achieve a Toffoli gate count of $\mathcal{O}(n^{\log_2(3)})$.
Our method offers improved parallelization of the circuit when using ancillary qubits.
In particular, we demonstrate that a linear depth can be achieved with $\mathcal{O}(n\log_2(n))$ ancillary qubits.
This results in a quantum circuit for binary field multiplication that provides the best known space-time trade-off while maintaining the subquadratic number of Toffoli gates.
Moreover, we show that for certain families of primitive polynomials, namely trinomials and equally spaced primitive polynomials, a logarithmic depth can be achieved using $\mathcal{O}(n^{\log_2(3)})$ ancillary qubits.
This leads to an even better space-time cost for these specific cases.
Table~\ref{tab:gf_comparison} summarizes the cost of different methods over various metrics.

In Section~\ref{sec:preliminaries}, we present the necessary background on quantum circuits and binary field multiplication.
Then, in Section~\ref{sec:gf_mult}, we detail our algorithm for binary field multiplication with linear depth and subquadratic Toffoli gate count.
Finally, in Section~\ref{sec:gf_log_depth}, we demonstrate how to achieve logarithmic depth for trinomials and equally spaced primitive polynomials.
Our implementation of the proposed algorithm is publicly available~\cite{github_gf}.

\begin{table}[t]
    \centering
    \begin{tabular}{lcccc} 
        \toprule
        Reference & Toffoli-count & Qubit-count & Depth & Space-time cost \\
        \bottomrule
        \rowcolor{gray!20} \multicolumn{5}{c}{\textbf{Arbitrary primitive polynomials}} \\
        Ref.~\cite{cheung2008design} & $\mathcal{O}(n^2)$ & $\mathcal{O}(n)$ & $\mathcal{O}(n)$ & $\mathcal{O}(n^2)$ \\
        Ref.~\cite{Rotteler_2014} & $\mathcal{O}(n^2)$ & $\mathcal{O}(n^2)$ & $\mathcal{O}(\log_2(n))$ & $\mathcal{O}(n^2\log_2(n))$ \\
        Ref.~\cite{Kepley_2015} & $\mathcal{O}(n^{\log_2(3)})$ & $\mathcal{O}(n^{\log_2(3)})$ & $\mathcal{O}(n)$ & $\mathcal{O}(n^{2.58})$\\
        Ref.~\cite{van_Hoof_2020, Putranto_2023} & $\mathcal{O}(n^{\log_2(3)})$ & $\mathcal{O}(n)$ & $\mathcal{O}(n^{\log_2(3)})$ & $\mathcal{O}(n^{2.58})$ \\
        Section~\ref{sec:gf_mult} & $\mathcal{O}(n^{\log_2(3)})$ & $\mathcal{O}(n\log_2(n))$ & $\mathcal{O}(n)$ & $\mathcal{O}(n^{2}\log_2(n))$\\
        \bottomrule
        \rowcolor{gray!20} \multicolumn{5}{c}{\textbf{Trinomials and equally spaced primitive polynomials}} \\
        \addlinespace[0.3em]
        Section~\ref{sec:gf_log_depth} & $\mathcal{O}(n^{\log_2(3)})$ & $\mathcal{O}(n^{\log_2(3)})$ & $\mathcal{O}(\log_2(n))$ & $\mathcal{O}(n^{1.58}\log_2(n))$\\
        \bottomrule
    \end{tabular}
    \caption{Toffoli-count, qubit-count, depth and space-time cost comparison between different quantum circuits performing multiplication over $GF(2^n)$.}\label{tab:gf_comparison}
\end{table}

\section{Preliminaries}\label{sec:preliminaries}
\subsection{Quantum gates}\label{sub:quantum_gates}

In this work, we will mostly make use of four quantum gates: the CNOT gate, the Toffoli gate, the CCZ gate and the Hadamard gate.
The representations of these gates in the quantum circuit model are illustrated in Figure~\ref{fig:quantum_gates}.

The Controlled-NOT (CNOT) gate is a two-qubit gate that performs a conditional NOT operation.
It applies the following transformation:
\begin{equation}
    \ket{a}\ket{b} \mapsto \ket{a}\ket{a \oplus b},
\end{equation}
where $a$ is the control qubit, $b$ is the target qubit, and $\oplus$ denotes the exclusive OR (XOR) operation.
If the control qubit $a$ is in the $\ket{1}$ state, the target qubit $b$ is flipped; otherwise, it remains unchanged.

The Toffoli gate, also known as the Controlled-Controlled-NOT (CCNOT) gate, is a three-qubit gate that extends the concept of the CNOT gate by using an additional control qubit.
When applied to qubits $a$, $b$, and $c$, where $c$ is the target qubit, it performs the following operation:
\begin{equation}
    \ket{a}\ket{b}\ket{c} \mapsto \ket{a}\ket{b}\ket{c \oplus ab}.
\end{equation}
The target qubit $c$ is flipped when both control qubits $a$ and $b$ are in the $\ket{1}$ state.

The Controlled-Controlled-Z (CCZ) gate is a three-qubit gate similar to the Toffoli gate, but it applies a phase shift instead of a bit flip.
It performs the following transformation:
\begin{equation}
    \ket{a}\ket{b}\ket{c} \mapsto (-1)^{abc}\ket{a}\ket{b}\ket{c}.
\end{equation}
A phase of $-1$ is applied when all three qubits are in the $\ket{1}$ state.

The Hadamard gate is a single-qubit gate that performs a basis transformation between the computational basis ($\{\ket{0}, \ket{1}\}$) and the Hadamard basis ($\{\ket{+}, \ket{-}\}$), where:
\begin{equation}
    \ket{+} = \frac{1}{\sqrt{2}}(\ket{0} + \ket{1}), \quad
    \ket{-} = \frac{1}{\sqrt{2}}(\ket{0} - \ket{1}).
\end{equation}
The basis transformation of the Hadamard gate can be used to described the relation between the CCZ gate and the Toffoli gate via the following circuit equality:
\begin{equation}\label{eq:ccz_toffoli}
    \includegraphics[valign=c]{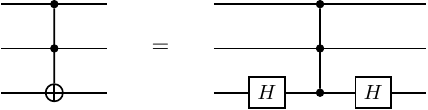}
\end{equation}

\begin{figure}[t]
    \centering
    \begin{subfigure}{0.24\textwidth} \centering
        \includegraphics{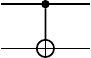}
        \caption{CNOT gate.}
    \end{subfigure}
    \begin{subfigure}{0.24\textwidth} \centering
        \includegraphics{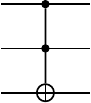}
        \caption{Toffoli gate.}
    \end{subfigure}
    \begin{subfigure}{0.24\textwidth} \centering
        \includegraphics{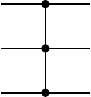}
        \caption{CCZ gate.}
    \end{subfigure}
    \begin{subfigure}{0.24\textwidth} \centering
        \includegraphics{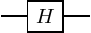}
        \caption{Hadamard gate.}
    \end{subfigure}
    \caption{Representation of quantum gates in the quantum circuit model.}
    \label{fig:quantum_gates}
\end{figure}

\subsection{Polynomial representation of CCZ circuits}\label{sub:ccz_circuits}
Let $C$ be a circuit composed only of CCZ gates operating on $n$ qubits.
From the definition of the CCZ gate, we can deduce that the action of $C$ has the form:
\begin{equation}
    \lvert \bs x\rangle \mapsto (-1)^{f(\bs x)} \lvert \bs x \rangle
\end{equation}
where $f$ is a third-order multilinear homogeneous polynomial with binary coefficients:
\begin{equation}
    f(\bs x) = \sum_{\alpha < \beta < \gamma}^n c_{\alpha, \beta, \gamma} x_\alpha x_\beta x_\gamma \pmod{2}
\end{equation}
where $c_{\alpha, \beta, \gamma} \in \mathbb{Z}_2$.
In the worst case, a CCZ circuit implementing the transformation described by $f(\bs x)$ would require $\mathcal{O}(n^3)$ CCZ gates, as there are potentially $\mathcal{O}(n^3)$ non-zero terms in $f$.
However, the number of CCZ gates can be optimized by inserting CNOT gates in the circuit.
For example, let $f$ be the following third-order multilinear homogeneous polynomial acting on 4 qubits:
\begin{equation}
    f(\bs x) = x_1x_3x_4 \oplus x_2x_3x_4
\end{equation}
A naive implementation would use two CCZ gates, one for each term.
However, we can use the tranformation realized by the CNOT gate in order to perform a change of basis and reduce the number of cubic terms in the polynomial.
In this example, this change of basis corresponds to factoring the polynomial as follows:
\begin{equation}
    f(\bs x) = x_1x_3x_4 \oplus x_2x_3x_4 = (x_1 \oplus x_2) x_3x_4
\end{equation}
By introducing a CNOT gate, we can create a new basis where one qubit represents the parity $(x_1 \oplus x_2)$.
This allows us to use only one CCZ gate instead of two.
Here are the initial circuit and the optimized circuit using CNOT gates to reduce the number of CCZ gates:
\begin{equation}
    \includegraphics[valign=c]{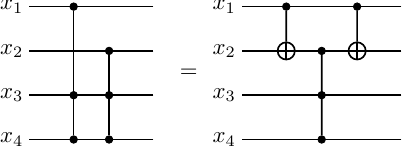}
\end{equation}
In Section~\ref{sec:gf_mult}, we will rely on this optimization technique using CNOT gates to construct an efficient circuit for multiplication over the field $GF(2^n)$.

\subsection{Binary field multiplication}\label{sub:gf_mult}

An element of $GF(2^n)$ can be represented by a polynomial $A(x)$ with binary coefficients and degree less than $n$:
\begin{equation}
    A(x) = \sum_{i=0}^{n-1} a_i x^i
\end{equation}
where $a_i \in GF(2)$ for all $i$.
Such polynomial $A(x)$ can be encoded into a binary vector $\bs a$ of size $n$ representing its coefficients:
\begin{equation}
    \bs a =
    \begin{bmatrix}
        a_0 \\
        \vdots \\
        a_{n-1} \\
    \end{bmatrix}.
\end{equation}
The multiplication of two polynomials $A(x)$ and $B(x)$ in $GF(2^n)$ is done modulo an irreducible primitive polynomial $P(x)$ of degree $n$:
\begin{equation}\label{eq:primitive_polynomial}
    P(x) = \sum_{i=0}^{n-1} p_i x^i + x^n
\end{equation}
where $p_i \in GF(2)$ for all $i$.
As proven in Reference~\cite{reyhani2004low}, multiplication over the field $GF(2^n)$ corresponds to the following operation:
\begin{equation}
    \ket{\bs a}\ket{\bs b}\ket{\bs 0} \mapsto \ket{\bs a}\ket{\bs b}\ket{\bs d \oplus Q \bs e}
\end{equation}
where all three registers are of size $n$, $Q$ is a binary matrix of size $n \times n-1$, and where $\bs d$ and $\bs e$ are vectors defined as
\begin{equation}
    \begin{aligned}
        \bs d &= L \bs b \\
        \bs e &= U \bs b,
    \end{aligned}
\end{equation}
with $L$ and $U$ defined as
\begin{equation}
    L =
    \begin{bmatrix}
        a_0 & 0 & \ldots & 0 & 0 \\
        a_1 & a_0 & \ldots & 0 & 0 \\
        \vdots & \vdots & \ddots & \vdots & \vdots \\
        a_{n-2} & a_{n-3} & \ldots & a_0 & 0 \\
        a_{n-1} & a_{n-2} & \ldots & a_1 & a_0 \\
    \end{bmatrix},\quad
    U =
    \begin{bmatrix}
        0 & a_{n-1} & a_{n-2} & \ldots & a_2 & a_1 \\
        0 & 0 & a_{n-1} & \ldots & a_3 & a_2 \\
        \vdots & \vdots & \vdots & \ddots & \vdots & \vdots \\
        0 & 0 & 0 & \ldots & a_{n-1} & a_{n-2} \\
        0 & 0 & 0 & \ldots & 0 & a_{n-1} \\
    \end{bmatrix}.
\end{equation}
The matrix $Q$, referred to as the reduction matrix, is derived from the primitive polynomial $P(x)$.
The key property of $Q$ is that it relates the higher-degree terms in the product to the lower-degree terms through the reduction modulo $P(x)$.
This reduction matrix satisfies the following equation:
\begin{equation}
    \bs x^\uparrow \equiv Q^T \bs x^\downarrow \pmod{P(x)}
\end{equation}
where $\bs x^\downarrow = [1, x, x^2, \ldots, x^{n-1}]^T$ and $\bs x^\uparrow = [x^n, x^{n+1}, \ldots, x^{2n-2}]^T$.
Using the fact that $P(x) \equiv 0 \pmod{P(x)}$, we can deduce from Equation~\ref{eq:primitive_polynomial} that
\begin{equation}
    x^n \equiv \sum_{i=0}^{n-1} p_i x^i \pmod{P(x)}.
\end{equation}
Hence, the first column of the reduction matrix $Q$ represents the coefficients $p_i$ of the primitive polynomial $P(x)$.
The subsequent columns of $Q$ can be easily computed by induction as $x^{n+i} \equiv x \cdot x^{n+i-1} \pmod{P(x)}$, for all integers $i > 0$.\\

\noindent\textbf{Example.} As an example, for the field $G(2^7)$ with primitive polynomial $P(x) = x^7 + x^5 + x^3 + x + 1$, the reduction matrix $Q$ is
\begin{equation}
    Q =
    \begin{bmatrix}
        1 & 0 & 1 & 0 & 0 & 0 \\
        1 & 1 & 1 & 1 & 0 & 0 \\
        0 & 1 & 1 & 1 & 1 & 0 \\
        1 & 0 & 0 & 1 & 1 & 1 \\
        0 & 1 & 0 & 0 & 1 & 1 \\
        1 & 0 & 0 & 0 & 0 & 1 \\
        0 & 1 & 0 & 0 & 0 & 0 \\
    \end{bmatrix}.
\end{equation}
The multiplication of the polynomials $A(x) = x^5 + x^3 + 1$ and $B(x) = x^2 + x$ with this primitive polynomial $P(x)$ is
\begin{equation}\label{eq:gf_mult_example}
    \begin{aligned}
        (x^5 + x^3 + 1)(x^2 + x) \pmod{P(x)}
        &= x^7 + x^5 + x^2 + x^6 + x^4 + x \pmod{P(x)} \\
        &= x^6 + x^4 + x^3 + x^2 + 1 \pmod{P(x)}
    \end{aligned}
\end{equation}
The binary vectors $\bs a, \bs b$ representing the polynomials $A(x)$ and $B(x)$ respectively, and the matrix $L$ and $U$ are:
\begin{equation}
    \bs a =
    \begin{bmatrix}
        1 \\
        0 \\
        0 \\
        1 \\
        0 \\
        1 \\
        0 \\
    \end{bmatrix},\quad
    \bs b =
    \begin{bmatrix}
        0 \\
        1 \\
        1 \\
        0 \\
        0 \\
        0 \\
        0 \\
    \end{bmatrix},\quad
    L =
    \begin{bmatrix}
        1 & 0 & 0 & 0 & 0 & 0 & 0 \\
        0 & 1 & 0 & 0 & 0 & 0 & 0 \\
        0 & 0 & 1 & 0 & 0 & 0 & 0 \\
        1 & 0 & 0 & 1 & 0 & 0 & 0 \\
        0 & 1 & 0 & 0 & 1 & 0 & 0 \\
        1 & 0 & 1 & 0 & 0 & 1 & 0 \\
        0 & 1 & 0 & 1 & 0 & 1 & 1 \\
    \end{bmatrix},\quad
    U =
    \begin{bmatrix}
        0 & 0 & 1 & 0 & 1 & 0 & 0 \\
        0 & 0 & 0 & 1 & 0 & 1 & 0 \\
        0 & 0 & 0 & 0 & 1 & 0 & 1 \\
        0 & 0 & 0 & 0 & 0 & 1 & 0 \\
        0 & 0 & 0 & 0 & 0 & 0 & 1 \\
        0 & 0 & 0 & 0 & 0 & 0 & 0 \\
    \end{bmatrix}.
\end{equation}
This leads to the following values for the vectors $\bs e$ and $\bs d$:
\begin{equation}
    \bs d = L \bs b = 
    \begin{bmatrix}
        0 & 1 & 1 & 0 & 1 & 1 & 1 \\
    \end{bmatrix}^T,\quad
    \bs e = U \bs b = 
    \begin{bmatrix}
        1 & 0 & 0 & 0 & 0 & 0 \\
    \end{bmatrix}^T.
\end{equation}
Finally, we have:
\begin{equation}
    \bs d \oplus Q \bs e = 
    \begin{bmatrix}
        1 & 0 & 1 & 1 & 1 & 0 & 1 \\
    \end{bmatrix}^T
\end{equation}
which represents the polynomial $x^6 + x^4 + x^3 + x^2 + 1$, matching the result obtained earlier in Equation~\ref{eq:gf_mult_example}.

\section{Synthesis algorithm for binary field multiplication}\label{sec:gf_mult}
In this section, we present methods for synthesizing binary field multiplier circuit using a subquadratic number of Toffoli gates.
First, in Subsection~\ref{sub:gf_subquadratic}, we introduce an algorithm for constructing a circuit that performs multiplication over $GF(2^n)$ using $\mathcal{O}(n^{\log_2(3)})$ Toffoli gates, without requiring any ancillary qubits.
Then, in Subsection~\ref{sub:gf_linear_depth}, we demonstrate how our approach can be adapted to construct a circuit for multiplication over $GF(2^n)$ with linear depth by using $\mathcal{O}(n\log_2(n))$ ancillary qubits, while maintaining the same subquadratic number of Toffoli gates.

\subsection{Subquadratic Toffoli gate count}\label{sub:gf_subquadratic}

Reference~\cite{cheung2008design} outlines a straightforward approach for implementing the following operation:
\begin{equation}\label{eq:gf_map}
    \ket{\bs a}\ket{\bs b}\ket{\bs 0} \mapsto \ket{\bs a}\ket{\bs b}\ket{\bs d \oplus Q \bs e}
\end{equation}
which corresponds to multiplication over $GF(2^n)$, as described in Section~\ref{sub:gf_mult}.
This method consists of three stages:
\begin{enumerate}
    \item First, the transformation $\ket{\bs a}\ket{\bs b}\ket{\bs 0} \mapsto \ket{\bs a}\ket{\bs b}\ket{\bs e}$ is achieved using Toffoli gates, where each gate has controls on the registers $\bs a$ and $\bs b$, and targets the third register.
    \item Then, the transformation $\ket{\bs a}\ket{\bs b}\ket{\bs e} \mapsto \ket{\bs a}\ket{\bs b}\ket{Q\bs e}$ is realized via a CNOT circuit on the third register.
    \item Similarly to the first step, the final transformation $\ket{\bs a}\ket{\bs b}\ket{\bs d \oplus Q\bs e} \mapsto \ket{\bs a}\ket{\bs b}\ket{\bs d \oplus M \bs e}$ is achieved using Toffoli gates, where each gate is has controls on the registers $\bs a$ and $\bs b$, and targets the third register.
\end{enumerate}
This method yields a circuit containing $\mathcal{O}(n^2)$ Toffoli gates for the first and last steps, and $\mathcal{O}(n^2)$ CNOT gates for the middle CNOT circuit.
This results in a circuit having a better Toffoli gate count than the naive approach using $\mathcal{O}(n^3)$ Toffoli gates.

To achieve a subquadratic Toffoli gate count, we will formalize the problem as a CCZ gate count optimization problem over the $\{\mathrm{CNOT}, CCZ\}$ gate set.
We can notice that the mapping of Equation~\ref{eq:gf_map} can be realized using only Toffoli gates, with targets on the third register and controls on the first two registers $\bs a$ and $\bs b$.
Then, by using the circuit equality of Equation~\ref{eq:ccz_toffoli}, we can conjugate all the Toffoli gates by a Hadamard gate on their target qubit.
Adjacent pairs of Hadamard can then be removed from the circuit by using the fact that the Hadamard gate is self-inverse.
This results in a circuit composed of a layer of Hadamard gates on the third register, followed by a layer of CCZ gates associated with $\bs d$ and $Q \bs e$, and then another layer of Hadamard gates on the third register.
As explained in Section~\ref{sub:ccz_circuits}, the CCZ subcircuit can then be represented by a third-order multilinear homogeneous polynomial with binary coefficients, noted $f$, satisfying
\begin{equation}\label{eq:gf_g_plus_h}
    f(\bs a, \bs b, \bs c, \bs c') = g(\bs a, \bs b, \bs c) \oplus h(\bs a, \bs b, \bs c')
\end{equation}
where $\bs a, \bs b, \bs c$ are the three registers and $\bs c' = Q^{-1} \bs c$, and where $g$ and $h$ are third-order multilinear homogeneous polynomials respectively associated with $\bs d$ and $Q\bs e$ and defined as follows:
\begin{align}
    g(\bs a, \bs b, \bs c) &= \sum_{i=0}^{n-1} \sum_{j=0}^{i} a_j b_{i - j} c_i \pmod{2} \label{eq:gf_g_def}, \\
    h(\bs a, \bs b, \bs c') &= \sum_{i=0}^{n-2} \sum_{j=i+1}^{n-1} a_j b_{n + i - j} c'_i \pmod{2}. \label{eq:gf_h_def}
\end{align}

Finding an implementation of the $GF(2^n)$ multiplier with a low number of CCZ (or Toffoli) gates then consists in performing the synthesis of the polynomial $f$ with an optimized number of CCZ gates.
To do so, we will rely on the following theorem, which defines a recursive formula for computing Equation~\ref{eq:gf_g_plus_h}.
This recursive formula have some similarities with the Karatsuba method for performing fast regular multiplications~\cite{karatsuba1962multiplication}.

\begin{theorem}\label{thm:gf_mult}
    Let $\bs a, \bs b, \bs c$ and $\bs c'$ be vectors of size $n$, where $n$ is an even integer.
    Then,
    \begin{equation}\label{eq:gf_thm}
        \begin{aligned}
            g(\bs a, \bs b, \bs c) \oplus h(\bs a, \bs b, \bs c') =\ &g(\bs a_L \oplus \bs a_R, \bs b_L \oplus \bs b_R, \bs c_R) \oplus h(\bs a_L \oplus \bs a_R, \bs b_L \oplus \bs b_R, \bs c'_L)\\
            &\oplus g(\bs a_R, \bs b_R, \bs c'_L \oplus \bs c_R) \oplus h(\bs a_R, \bs b_R, \bs c'_L \oplus \bs c'_R) \\
            &\oplus g(\bs a_L, \bs b_L, \bs c_L \oplus \bs c_R) \oplus h(\bs a_L, \bs b_L, \bs c'_L \oplus \bs c_R)
        \end{aligned}
    \end{equation}
    where the vectors $\bs a_L$ and $\bs a_R$ are respectively equal to the left and right halves of the vector $\bs a$.
\end{theorem}

The proof of Theorem~\ref{thm:gf_mult} is provided in Appendix~\ref{app:gf_mult}.
For the base case where $n = 1$, we can easily deduce from Equations~\ref{eq:gf_g_def} and~\ref{eq:gf_h_def} that
\begin{equation}\label{eq:gf_base_case}
    \begin{aligned}
        g(\bs a, \bs b, \bs c) \oplus h(\bs a, \bs b, \bs c') = a_0 b_0 c_0 \\
    \end{aligned}
\end{equation}
where $\bs a$, $\bs b$, $\bs c$ and $\bs c'$ are vectors of size $1$, which corresponds to a single CCZ gate on the qubits $a_0, b_0$ and $c_0$.
To use the formula given in Theorem~\ref{thm:gf_mult}, the size of the vectors, noted $n$, must be even. 
The following theorem demonstrates how we can still rely on the recursive formula of Equation~\ref{eq:gf_thm} when $n$ is odd by increasing the size of the vectors by $1$.

\begin{theorem}\label{thm:gf_mult_odd}
    Let $\bs a, \bs b, \bs c$ and $\bs c'$ be vectors of size $n$.
    And let $\tilde{\bs a}, \tilde{\bs b}, \tilde{\bs c}$ and $\tilde{\bs c}'$ be vectors of size $n+1$, defined as follows:
    \begin{equation}
    \begin{aligned}
        \tilde{\bs a} &= \begin{bmatrix} \bs a \\ 0 \end{bmatrix},\quad
        \tilde{\bs b} &= \begin{bmatrix} \bs b \\ 0 \end{bmatrix}, \quad
        \tilde{\bs c} &= \begin{bmatrix} \bs c \\ c'_0 \end{bmatrix}, \quad
        \tilde{\bs c}' &= \begin{bmatrix} c'_1 \\ \vdots \\ c'_{n-1} \\ 0 \\ 0 \end{bmatrix}.
    \end{aligned}
    \end{equation}
    Then,
    \begin{equation}
    \begin{aligned}
        g(\bs a, \bs b, \bs c) \oplus  h(\bs a, \bs b, \bs c') = g(\tilde{\bs a}, \tilde{\bs b}, \tilde{\bs c}) \oplus  h(\tilde{\bs a}, \tilde{\bs b}, \tilde{\bs c}'). \\
    \end{aligned}
    \end{equation}
\end{theorem}
The proof of Theorem~\ref{thm:gf_mult_odd} is provided in Appendix~\ref{app:gf_mult_odd}.\\

\noindent\textbf{Complexity analysis.}
Theorem~\ref{thm:gf_mult} divides the sum of $g$ and $h$ with vectors of size $n$ into $3$ smaller sums, each involving vectors of size $n/2$.
Thus, by applying the equality of Theorem~\ref{thm:gf_mult} recursively, the sum of $g$ and $h$ with vectors of size $n$ will be divided into $3^{\lceil\log_2(n)\rceil}$ sums of $g$ and $h$ with vectors of single elements, such as in Equation~\ref{eq:gf_base_case}.
Therefore, the constructed circuit performs multiplication over the field $GF(2^n)$ with no more than $3^{\lceil\log_2(n)\rceil}$ CCZ gates.
This correponds to $\mathcal{O}(n^{\log_2(3)}) \approx \mathcal{O}(n^{1.58})$ CCZ gates, without requiring any ancillary qubits.

\subsection{Linear depth}\label{sub:gf_linear_depth}

In this section, we demonstrate how Theorem~\ref{thm:gf_mult} can be utilized to construct a quantum circuit that performs multiplication over the field $GF(2^n)$ with linear depth by using $\mathcal{O}(n\log_2(n))$ ancillary qubits.

To parallelize the circuit, we begin by using $n$ ancillary qubits to store the register $\bs c' = Q^{-1} \bs c$:
\begin{equation}\label{eq:c_prime_map}
    \ket{\bs a}\ket{\bs b}\ket{\bs c}\ket{\bs 0} \mapsto \ket{\bs a}\ket{\bs b}\ket{\bs c}\ket{\bs c'}.
\end{equation}
Then, to apply the recursive formula of Theorem~\ref{thm:gf_mult}, we perform the following transformation:
\begin{equation}\label{eq:gf_map_par}
    \ket{\bs a}\ket{\bs b}\ket{\bs c}\ket{\bs c'}\ket{\bs 0} \mapsto \ket{\bs a_L \oplus \bs a_R}\ket{\bs a_R}\ket{\bs b_L \oplus \bs b_R}\ket{\bs b_R}\ket{\bs c_L}\ket{\bs c_R}\ket{\bs c'_L}\ket{\bs c'_L \oplus \bs c'_R}\ket{\bs c'_L \oplus \bs c_R}.
\end{equation}
This transformation can be done with a CNOT circuit of depth 2.
The first layer consists of adding $\bs a_R$ and $\bs b_R$ to $\bs a_L$ and $\bs b_L$ respectively, as well as adding $\bs c'_L$ to $\bs c'_R$ and $\bs c_R$ to the clean ancillae register.
The second layer of CNOT gates simply consists of adding $\bs c'_L$ to the last register storing $\bs c_R$.
This preparation enables the parallel execution of the first two recursive calls of the formula presented in Theorem~\ref{thm:gf_mult}:
\begin{equation}
    g(\bs a_L \oplus \bs a_R, \bs b_L \oplus \bs b_R, \bs c_R) \oplus h(\bs a_L \oplus \bs a_R, \bs b_L \oplus \bs b_R, \bs c'_L)
\end{equation}
and
\begin{equation}
    g(\bs a_R, \bs b_R, \bs c'_L \oplus \bs c_R) \oplus h(\bs a_R, \bs b_R, \bs c'_L \oplus \bs c'_R).
\end{equation}
These recursive calls can be executed simultaneously as they operate on independent registers.
After these recursive calls, we reverse the transformation of Equation~\ref{eq:gf_map_par} using the inverse of the CNOT circuit, which is also done in constant depth.
Finally, we realize the following transformation:
\begin{equation}
    \ket{\bs a}\ket{\bs b}\ket{\bs c}\ket{\bs c'}\ket{\bs 0} \mapsto \ket{\bs a_L}\ket{\bs a_R}\ket{\bs b_L}\ket{\bs b_R}\ket{\bs c_L \oplus \bs c_R}\ket{\bs c_R}\ket{\bs c'_L \oplus \bs c_R}\ket{\bs c'_R}
\end{equation}
This can also be accomplished with a CNOT circuit of depth 2 by first adding $\bs c_R$ to $\bs c_L$ and then adding $\bs c_R$ to $\bs c'_L$.
This transformation allows for the execution of the third recursive call from Theorem~\ref{thm:gf_mult}:
\begin{equation}
    g(\bs a_L, \bs b_L, \bs c_L \oplus \bs c_R) \oplus h(\bs a_L, \bs b_L, \bs c'_L \oplus \bs c_R).
\end{equation}
An illustrative example of the circuit generated by this procedure for $GF(2^4)$ is provided in Appendix~\ref{app:gf_example}.\\

\noindent\textbf{Complexity analysis.}
The transformation described in Equation~\ref{eq:c_prime_map} can be implemented by CNOT circuit with a depth of $\mathcal{O}(n)$~\cite{Kutin_2007}.
Then, the formula of Theorem~\ref{thm:gf_mult} is applied to divide the sum of $g$ and $h$ into three recursive calls.
Two of these recursive calls are done in parallel by using $\lceil k/2 \rceil$ ancillary qubits, where $k$ is the size of the registers on which $g$ and $h$ are acting.
Each recursive call results in either a constant-depth CNOT circuit or a single CCZ gate for the terminal case.
Moreover, each recursive call divides the size of the registers by 2.
Therefore, the depth of the circuit implementing all these recursive calls is equal to $\mathcal{O}(2^{\log_2(n)}) = \mathcal{O}(n)$.
Thus, the total depth of the circuit is $\mathcal{O}(n)$ and the number of ancillary qubits in the circuit is $\mathcal{O}(n\log_2(n))$.

\section{Logarithmic depth multiplication for particular families of primitive polynomials}\label{sec:gf_log_depth}
In this section, we show how multiplication over $GF(2^n)$ can be done with a logarithmic depth in the cases where the primitive polynomial is a trinomial or an equally spaced polynomial.
To do so, we will show that the CNOT circuit associated with the reduction matrix for these cases can be constructed with a logarithmic depth.
In particular, we will rely on the CNOT ladder operator, denoted by $\mathcal{L}_1^{(n)}$ and acting on a $(n+1)$-dimensional computational basis vector $\ket{\bs x}$ as follows:
\begin{equation}
    \mathcal{L}_1^{(n)} \ket{\bs x} = \ket{x_0} \bigotimes_{i=1}^{n} \ket{x_i \oplus x_{i-1}}.
\end{equation}
It has been demonstrated in Reference~\cite{remaud2025} that this CNOT ladder operator can be implemented with a depth of $\mathcal{O}(\log_2(n))$ using only CNOT gates.

\subsection{Trinomials}
A primitive trinomial of degree $n$ satisfies:
\begin{equation}
    P(x) = x^n + x^{k} + 1
\end{equation}
where $k$ is an integer satisfying $1 < k < n$.
Let $i$ be an integer such that $0 \leq i < (n-k)$, and let $j$ be an integer such that $0 \leq i+j(n-k) \leq n-2$.
The $(i+j(n-k))$-th column of the associated reduction matrix $Q$ represents the coefficients of the polynomial
\begin{equation}\label{eq:trinomials}
    x^{k+i+j(n-k)} + x^{i+j(n-k)} \pmod{P(x)}.
\end{equation}

For $j = 0$, Equation~\ref{eq:trinomials} simplifies to
\begin{equation}
    x^{k+i} + x^{i} \pmod{P(x)}.
\end{equation}
The CNOT circuit associated with these $(n-k)$ columns can be constructed by applying parallel CNOT gates with the qubit $i$ as control and the qubit $i+k$ as target.

For $j > 0$, we have the following recursive relation:
\begin{equation}
    \begin{aligned}
        &x^{k+i+j(n-k)} + x^{i+j(n-k)} \pmod{P(x)} \\
        &= x^{n} x^{i+(j-1)(n-k)} + x^{i+j(n-k)} \pmod{P(x)} \\
        &= (x^{k} + 1) x^{i+(j-1)(n-k)} + x^{i+j(n-k)} \pmod{P(x)} \\
        &= x^{k + i+(j-1)(n-k)} + x^{i+(j-1)(n-k)} + x^{i+j(n-k)} \pmod{P(x)}.
    \end{aligned}
\end{equation}
Therefore, the $(i+j(n-k))$-th column of $Q$ is equal to the $(i+(j-1)(n-k))$-th column of $Q$ plus an additional $1$ on the diagonal (corresponding to the term $x^{i+j(n-k)}$).
The CNOT circuit associated with these columns can then be constructed by applying $(n-k)$ ladders of CNOT gates in parallel.
As proven in Reference~\cite{remaud2025}, a CNOT ladder over $n$ qubits can be implemented with a depth of $\mathcal{O}(\log_2(n))$.
Thus, the CNOT circuit associated with the reduction matrix $Q$ for the primitive trinomial $P(x)$ can be constructed with a logarithmic depth.

For example, for the trinomial $P(x) = x^9 + x^7 + 1$, the associated reduction matrix $Q$ is
\begin{equation}
    Q =
    \begin{bmatrix}
        1 & 0 & 1 & 0 & 1 & 0 & 1 & 0 \\
        0 & 1 & 0 & 1 & 0 & 1 & 0 & 1 \\
        0 & 0 & 1 & 0 & 1 & 0 & 1 & 0 \\
        0 & 0 & 0 & 1 & 0 & 1 & 0 & 1 \\
        0 & 0 & 0 & 0 & 1 & 0 & 1 & 0 \\
        0 & 0 & 0 & 0 & 0 & 1 & 0 & 1 \\
        0 & 0 & 0 & 0 & 0 & 0 & 1 & 0 \\
        1 & 0 & 1 & 0 & 1 & 0 & 1 & 1 \\
        0 & 1 & 0 & 1 & 0 & 1 & 0 & 1 \\
    \end{bmatrix}.
\end{equation}
A CNOT circuit associated with this reduction matrix is:
\begin{center}
    \includegraphics{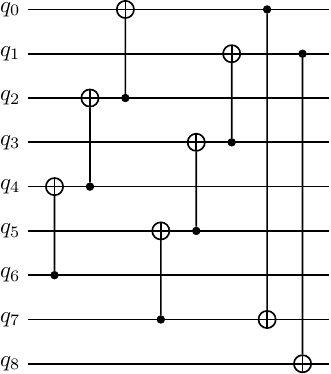}
\end{center}
where the two ladders of CNOT gates can be replaced by a logarithmic-depth CNOT circuit~\cite{remaud2025}.

\subsection{Equally spaced polynomials}
An equally spaced primitive polynomial $P(x)$ of degree $nk$ satisfies
\begin{equation}
    P(x) = \sum_{i=0}^n x^{ik}.
\end{equation}
where $k$ is an integer satisfying $0 < k < n$.
Let $j$ be an integer such that $0 \leq i < k$.
The $j$-th column of the associated reduction matrix $Q$ represents the coefficients of the polynomial
\begin{equation}
    \sum_{i=0}^{n-1} x^{ik+j} \pmod{P(x)}.
\end{equation}
The CNOT circuit associated with these $k$ columns can be constructed by applying $k$ parallel ladders of CNOT gates on the qubits $(j, k+j, 2k+j, \ldots, (n-1)k+j)$.

All the other columns $j + \ell k$, where $\ell$ is an integer satisfying $0 < \ell < n$, are representing the coefficients of the polynomial
\begin{equation}
    \begin{aligned}
        \sum_{i=0}^{n-1} x^{ik+j+\ell k} \pmod{P(x)} &= \sum_{i=0}^{n-1} x^{(i+\ell)k+j} \pmod{P(x)} \\
        &= x^{(n+\ell-1)k+j}+\sum_{i=1}^{n-1} x^{(i+\ell-1)k+j} \pmod{P(x)} \\
        &= x^{nk}x^{(\ell-1)k+j}+x^{(\ell-1)k+j}\sum_{i=1}^{n-1}x^{ik} \pmod{P(x)} \\
        &= x^{(\ell-1)k+j}\left(\sum_{i=0}^{n-1}x^{ik} + \sum_{i=1}^{n-1}x^{ik}\right) \pmod{P(x)} \\
        &= x^{(\ell-1)k+j} \pmod{P(x)}. \\
    \end{aligned}
\end{equation}
The CNOT circuit associated with these $n(k-1)-1$ columns can be constructed by applying $k$ ladder of CNOT gates on the qubits $(j, k+j, 2k+j, \ldots, (n-1)k+j)$.
As proven in Reference~\cite{remaud2025}, a CNOT ladder over $n$ qubits can be implemented with a depth of $\mathcal{O}(\log_2(n))$.
Thus, the CNOT circuit associated with the reduction matrix $Q$ for the equally spaced polynomial $P(x)$ can be constructed with a logarithmic depth.

For example, for $n=4$ and $k=2$, we have the equally spaced polynomial $P(x) = x^8 + x^6 + x^4 + x^2 + 1$.
Its associated reduction matrix $Q$ is
\begin{equation}
    Q =
    \begin{bmatrix}
        1 & 0 & 1 & 0 & 0 & 0 & 0 \\
        0 & 1 & 0 & 1 & 0 & 0 & 0 \\
        1 & 0 & 0 & 0 & 1 & 0 & 0 \\
        0 & 1 & 0 & 0 & 0 & 1 & 0 \\
        1 & 0 & 0 & 0 & 0 & 0 & 1 \\
        0 & 1 & 0 & 0 & 0 & 0 & 0 \\
        1 & 0 & 0 & 0 & 0 & 0 & 0 \\
        0 & 1 & 0 & 0 & 0 & 0 & 0 \\
    \end{bmatrix}.
\end{equation}
A CNOT circuit associated with this reduction matrix is:
\begin{center}
    \includegraphics{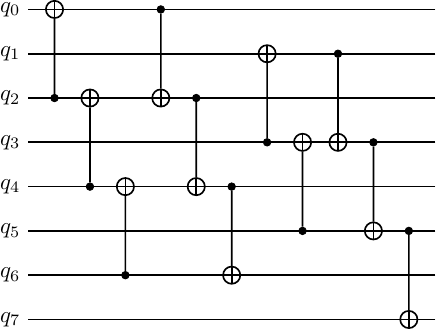}
\end{center}
where the $2k$ ladders of CNOT gates can be replaced by a logarithmic-depth CNOT circuit~\cite{remaud2025}.

\section{Conclusion}

We presented a novel approach for constructing efficient quantum circuits that perform binary field multiplication, a fundamental operation in quantum cryptanalysis of binary elliptic curve cryptography.
Our method achieves a Toffoli gate count of $\mathcal{O}(n^{\log_2(3)})$, matching the best known results, while offering a significant improvement in terms of space-time cost.
More specifically, the circuits generated by our algorithm have a linear depth and use $\mathcal{O}(n\log_2(n))$ ancillary qubits.
This corresponds to a space-time cost of $\mathcal{O}(n^2\log_2(n))$, which is the best-known space-time cost for quantum circuits implementing binary field multiplication with a subquadratic number of Toffoli gates.

An open question is whether this space-time cost can be further reduced.
We partially answered this question positively by demonstrating that there are indeed cases where further optimization is possible.
In particular, when the multiplication is performed modulo a primitive polynomial which is either a trinomial or an equally spaced polynomial, we showed that a logarithmic depth can be achieved by using $\mathcal{O}(n^{\log_2(3)})$ ancillary qubits.
An interesting direction for future work would be to investigate whether there exist other families of primitive polynomials for which multiplication can be performed in logarithmic depth.

\section*{Acknowledgments}
We acknowledge funding from the Plan France 2030 through the projects NISQ2LSQ ANR-22-PETQ-0006 and EPIQ ANR-22-PETQ-007.

\bibliographystyle{quantum}
\bibliography{ref.bib}

\newpage
\appendix

\section{Proof of Theorem~\ref{thm:gf_mult}}\label{app:gf_mult}

\begin{proof}[Proof of Theorem~\ref{thm:gf_mult}]
Equation~\ref{eq:gf_thm} holds if and only if the two following equations are true:
\begin{align}
        \!g(\bs a, \bs b, \bs c) &= g(\bs a_L \oplus \bs a_R, \bs b_L \oplus \bs b_R, \bs c_R) \oplus g(\bs a_R, \bs b_R, \bs c_R) \oplus g(\bs a_L, \bs b_L, \bs c_L \oplus \bs c_R) \oplus h(\bs a_L, \bs b_L, \bs c_R)  \label{eq:gf_2} \\
        \!h(\bs a, \bs b, \bs c') &= h(\bs a_L \oplus \bs a_R, \bs b_L \oplus \bs b_R, \bs c'_L) \oplus g(\bs a_R, \bs b_R, \bs c'_L) \oplus h(\bs a_R, \bs b_R, \bs c'_L \oplus \bs c'_R) \oplus h(\bs a_L, \bs b_L, \bs c'_L) \label{eq:gf_1}
\end{align}
We first prove Equation~\ref{eq:gf_2}:
\begin{equation}
    \begin{aligned}
        &g(\bs a_L \oplus \bs a_R, \bs b_L \oplus \bs b_R, \bs c_R) \oplus g(\bs a_R, \bs b_R, \bs c_R) \oplus g(\bs a_L, \bs b_L, \bs c_L \oplus \bs c_R) \oplus h(\bs a_L, \bs b_L, \bs c_R) \\
        &= \sum_{i=0}^{\frac{n}{2}-1} \sum_{j=0}^{i} (a_j + a_{\frac{n}{2} + j}) (b_{i - j} + b_{\frac{n}{2} + i - j}) c_{\frac{n}{2} + i} + \sum_{i=0}^{\frac{n}{2}-1} \sum_{j=0}^{i} a_{\frac{n}{2} + j} b_{\frac{n}{2} + i - j} c_{\frac{n}{2} + i} \\
        &\quad+ \sum_{i=0}^{\frac{n}{2}-1} \sum_{j=0}^{i} a_j b_{i - j} (c_i + c_{\frac{n}{2} + i}) + \sum_{i=0}^{\frac{n}{2}-2} \sum_{j=i+1}^{\frac{n}{2}-1} a_j b_{\frac{n}{2} + i - j} c_{\frac{n}{2} + i} \pmod{2} \\
        &= \sum_{i=0}^{\frac{n}{2}-1} \sum_{j=0}^{i} a_j b_{\frac{n}{2} + i - j} c_{\frac{n}{2} + i} + \sum_{i=0}^{\frac{n}{2}-1} \sum_{j=0}^{i} a_{\frac{n}{2} + j} b_{i - j} c_{\frac{n}{2} + i} + \sum_{i=0}^{\frac{n}{2}-1} \sum_{j=0}^{i} a_j b_{i - j} c_i \\
        &\quad+ \sum_{i=0}^{\frac{n}{2}-2} \sum_{j=i+1}^{\frac{n}{2}-1} a_j b_{\frac{n}{2} + i - j} c_{\frac{n}{2} + i} \pmod{2} \\
        &= \sum_{i=\frac{n}{2}}^{n-1} \sum_{j=0}^{i - \frac{n}{2}} a_j b_{i - j} c_{i} + \sum_{i=\frac{n}{2}}^{n-1} \sum_{j=\frac{n}{2}}^{i} a_{j} b_{i - j} c_{i} + \sum_{i=0}^{\frac{n}{2}-1} \sum_{j=0}^{i} a_j b_{i - j} c_i \\
        &\quad+ \sum_{i=\frac{n}{2}}^{n-2} \sum_{j=i-\frac{n}{2}+1}^{\frac{n}{2}-1} a_j b_{i - j} c_{i} \pmod{2} \\
          &= \sum_{i=\frac{n}{2}}^{n-1} \sum_{j=0}^{\frac{n}{2} - 1} a_j b_{i - j} c_{i} + \sum_{i=\frac{n}{2}}^{n-1} \sum_{j=\frac{n}{2}}^{i} a_{j} b_{i - j} c_{i} + \sum_{i=0}^{\frac{n}{2}-1} \sum_{j=0}^{i} a_j b_{i - j} c_i \pmod{2} \\
        &= \sum_{i=\frac{n}{2}}^{n-1} \sum_{j=0}^{i} a_j b_{i - j} c_{i} + \sum_{i=0}^{\frac{n}{2}-1} \sum_{j=0}^{i} a_j b_{i - j} c_i \pmod{2} \\
        &= \sum_{i=0}^{n-1} \sum_{j=0}^{i} a_j b_{i - j} c_i \pmod{2} \\
        &= g(\bs a, \bs b, \bs c).
    \end{aligned}
\end{equation}
We now prove Equation~\ref{eq:gf_1}:
\begin{equation}
    \begin{aligned}
        &h(\bs a_L \oplus \bs a_R, \bs b_L \oplus \bs b_R, \bs c'_L) \oplus g(\bs a_R, \bs b_R, \bs c'_L) \oplus h(\bs a_R, \bs b_R, \bs c'_L \oplus \bs c'_R) \oplus h(\bs a_L, \bs b_L, \bs c'_L) \\
        &= \sum_{i=0}^{\frac{n}{2}-2} \sum_{j=i+1}^{\frac{n}{2}-1} (a_j + a_{\frac{n}{2} + j}) (b_{\frac{n}{2} + i - j} + b_{n + i - j}) c'_i + \sum_{i=0}^{\frac{n}{2}-1} \sum_{j=0}^{i} a_{\frac{n}{2} + j} b_{\frac{n}{2} + i - j} c'_i \\
        &\quad+ \sum_{i=0}^{\frac{n}{2}-2} \sum_{j=i+1}^{\frac{n}{2}-1} a_{\frac{n}{2} + j} b_{n + i - j} (c'_i + c'_{\frac{n}{2} + i}) + \sum_{i=0}^{\frac{n}{2}-2} \sum_{j=i+1}^{\frac{n}{2}-1} a_j b_{\frac{n}{2} + i - j} c'_i \pmod{2} \\
        &= \sum_{i=0}^{\frac{n}{2}-2} \sum_{j=i+1}^{\frac{n}{2}-1} a_j b_{n + i - j} c'_i  + \sum_{i=0}^{\frac{n}{2}-2} \sum_{j=i+1}^{\frac{n}{2}-1} a_{\frac{n}{2} + j} b_{\frac{n}{2} + i - j} c'_i + \sum_{i=0}^{\frac{n}{2}-1} \sum_{j=0}^{i} a_{\frac{n}{2} + j} b_{\frac{n}{2} + i - j} c'_i \\
        &\quad+ \sum_{i=0}^{\frac{n}{2}-2} \sum_{j=i+1}^{\frac{n}{2}-1} a_{\frac{n}{2} + j} b_{n + i - j} c'_{\frac{n}{2} + i} \pmod{2} \\
        &= \sum_{i=0}^{\frac{n}{2}-2} \sum_{j=i+1}^{\frac{n}{2}-1} a_j b_{n + i - j} c'_i  + \sum_{i=0}^{\frac{n}{2}-2} \sum_{j=\frac{n}{2}+i+1}^{n-1} a_{j} b_{n + i - j} c'_i  + \sum_{i=0}^{\frac{n}{2}-1} \sum_{j=\frac{n}{2}}^{\frac{n}{2} + i} a_{j} b_{n + i - j} c'_i \\
        &\quad+ \sum_{i=\frac{n}{2}}^{n-2} \sum_{j=i+1}^{n-1} a_{j} b_{n + i - j} c'_{i} \pmod{2} \\
        &= \sum_{i=0}^{\frac{n}{2}-2} \sum_{j=i+1}^{\frac{n}{2}-1} a_j b_{n + i - j} c'_i  + \sum_{i=0}^{\frac{n}{2}-2} \sum_{j=\frac{n}{2}+i+1}^{n-1} a_{j} b_{n + i - j} c'_i + \sum_{i=0}^{\frac{n}{2}-2} \sum_{j=\frac{n}{2}}^{\frac{n}{2} + i} a_{j} b_{n + i - j} c'_i \\
        &\quad+  \sum_{i=\frac{n}{2}-1}^{\frac{n}{2}-1} \sum_{j=i+1}^{n-1} a_{j} b_{n + i - j} c'_i + \sum_{i=\frac{n}{2}}^{n-2} \sum_{j=i+1}^{n-1} a_{j} b_{n + i - j} c'_{i} \pmod{2} \\
        &= \sum_{i=0}^{\frac{n}{2}-2} \sum_{j=i+1}^{\frac{n}{2}-1} a_j b_{n + i - j} c'_i  + \sum_{i=0}^{\frac{n}{2}-2} \sum_{j=\frac{n}{2}+i+1}^{n-1} a_{j} b_{n + i - j} c'_i + \sum_{i=0}^{\frac{n}{2}-2} \sum_{j=\frac{n}{2}}^{\frac{n}{2} + i} a_{j} b_{n + i - j} c'_i \\
        &\quad+ \sum_{i=\frac{n}{2}-1}^{n-2} \sum_{j=i+1}^{n-1} a_{j} b_{n + i - j} c'_{i} \pmod{2} \\
        &= \sum_{i=0}^{\frac{n}{2}-2} \sum_{j=i+1}^{\frac{n}{2}+i} a_j b_{n + i - j} c'_i  + \sum_{i=0}^{\frac{n}{2}-2} \sum_{j=\frac{n}{2}+i+1}^{n-1} a_{j} b_{n + i - j} c'_i + \sum_{i=\frac{n}{2}-1}^{n-2} \sum_{j=i+1}^{n-1} a_{j} b_{n + i - j} c'_{i} \pmod{2} \\
        &= \sum_{i=0}^{\frac{n}{2}-2} \sum_{j=i+1}^{n-1} a_{j} b_{n + i - j} c'_i + \sum_{i=\frac{n}{2}-1}^{n-2} \sum_{j=i+1}^{n-1} a_{j} b_{n + i - j} c'_{i} \pmod{2} \\
        &= \sum_{i=0}^{n-2} \sum_{j=i+1}^{n-1} a_{j} b_{n + i - j} c'_{i} \pmod{2} \\
        & = h(\bs a, \bs b, \bs c').
    \end{aligned}
\end{equation}
\end{proof}

\section{Proof of Theorem~\ref{thm:gf_mult_odd}}\label{app:gf_mult_odd}
\begin{proof}[Proof of Theorem~\ref{thm:gf_mult_odd}]
    \begin{equation}
    \begin{aligned}
        &g(\tilde{\bs a}, \tilde{\bs b}, \tilde{\bs c}) \oplus h(\tilde{\bs a}, \tilde{\bs b}, \tilde{\bs c}') \\
        &= \sum_{i=0}^{n} \sum_{j=0}^{i} \tilde{a}_j \tilde{b}_{i - j} \tilde{c}_i + \sum_{i=0}^{n-1} \sum_{j=i+1}^{n} \tilde{a}_j \tilde{b}_{n + 1 + i - j} \tilde{c}'_i \pmod{2} \\
        &= \sum_{i=0}^{n-1} \sum_{j=0}^{i} a_j b_{i - j} c_i + \sum_{j=1}^{n-1} a_j b_{n - j} c'_0 + \sum_{i=0}^{n-3} \sum_{j=i+2}^{n-1} a_j b_{n + 1 + i - j} c'_{i+1} \pmod{2} \\
        &= g(\bs a, \bs b, \bs c) + \sum_{j=1}^{n-1} a_j b_{n - j} c'_0 + \sum_{i=1}^{n-2} \sum_{j=i+1}^{n-1} a_j b_{n + i - j} c'_i \pmod{2} \\
        &= g(\bs a, \bs b, \bs c) + \sum_{i=0}^{0} \sum_{j=i+1}^{n-1} a_j b_{n + i - j} c'_i + \sum_{i=1}^{n-2} \sum_{j=i+1}^{n-1} a_j b_{n + i - j} c'_i \pmod{2} \\
        &= g(\bs a, \bs b, \bs c) + \sum_{i=0}^{n-2} \sum_{j=i+1}^{n-1} a_j b_{n + i - j} c'_i \pmod{2} \\
        &= g(\bs a, \bs b, \bs c) \oplus h(\bs a, \bs b, \bs c')
    \end{aligned}
    \end{equation}
\end{proof}

\begin{figure}[t]
\section{Example}\label{app:gf_example}
\centering
\rotatebox{90}{%
  \begin{minipage}{\textheight}
    \centering
    \hspace*{0.3in}
    \includegraphics[width=0.95\columnwidth]{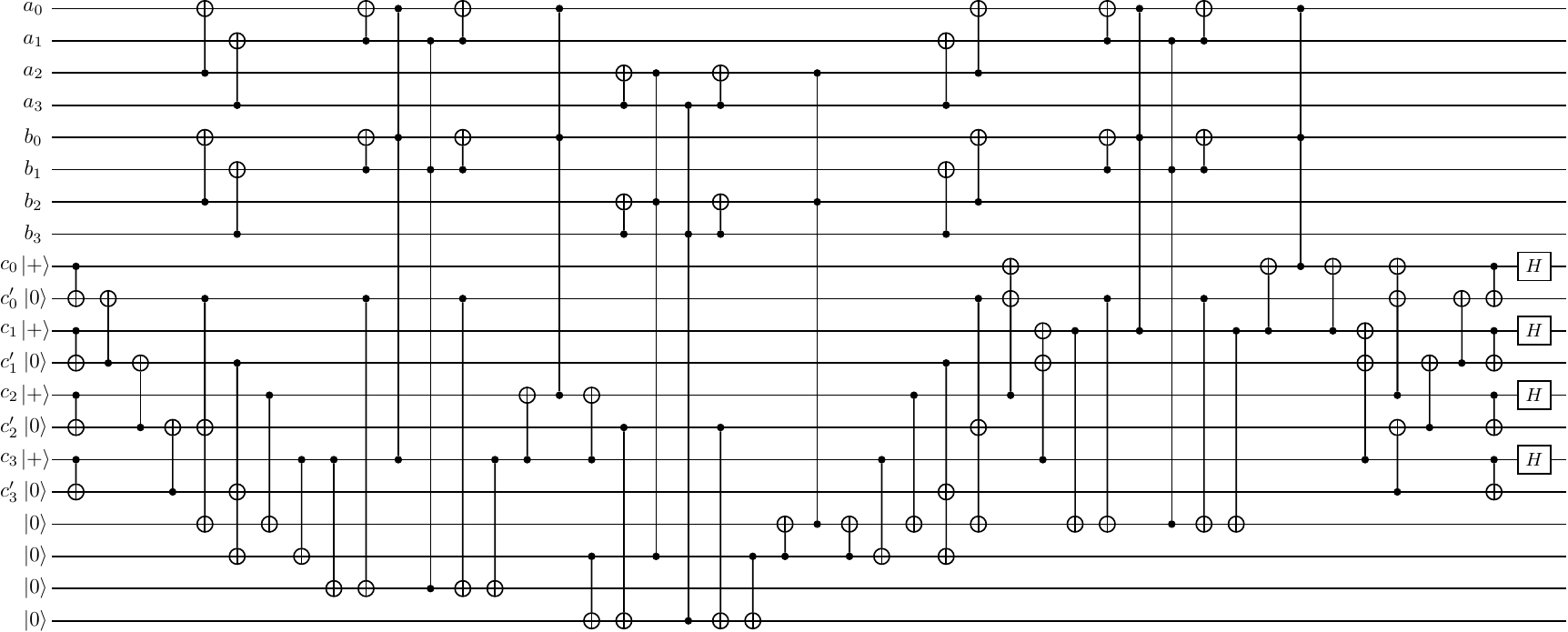}
    \captionof{figure}{Linear depth multiplication circuit for $GF(2^4)$ with primitive polynomial $P(x) = x^4 + x + 1$.}
    \label{fig:gf_mult_4}
  \end{minipage}%
}
\end{figure}
\end{document}